\documentclass[]{interact}

\usepackage[utf8]{inputenc}
\usepackage{graphicx}
\usepackage{amsfonts}
\usepackage{amsmath}
\usepackage{hyperref}
\usepackage{amsthm}
\usepackage{mathtools}
\usepackage{bbm}
\usepackage[numbers,sort&compress]{natbib}% Citation support using natbib.sty
\bibpunct[, ]{[}{]}{,}{n}{,}{,}% Citation support using natbib.sty
\graphicspath{{./Pics/}}

\newcommand{\numberset}{\mathbb}
\newcommand{\N}{\numberset{N}}
\newcommand{\R}{\numberset{R}}
\newcommand{\E}{\numberset{E}}

\newcommand{\PP}{\numberset{P}}
\newcommand{\Q}{\numberset{Q}}

\newcommand{\LL}{\mathcal{L}}

\theoremstyle{plain}% Theorem-like structures
\newtheorem{theorem}{Theorem}[section]

\theoremstyle{definition}

\theoremstyle{remark}
\newtheorem{remark}{Remark}

\begin{document}

\title{Multinomial method for option pricing under Variance Gamma}

\author{
\name{Nicola Cantarutti\textsuperscript{$\dagger$}\thanks{Nicola Cantarutti. Email: nicolacantarutti@gmail.com} and João Guerra\textsuperscript{$\dagger$}}
\affil{\textsuperscript{$\dagger$}CEMAPRE - Center for Applied Mathematics and Economics - ISEG - University of Lisbon.} }

\maketitle

\begin{abstract}
This paper presents a multinomial method for option pricing
when the underlying asset follows an exponential Variance Gamma process.
The continuous time Variance Gamma process is approximated by a continuous time process with the same first four cumulants, and then discretized in time and space. 
This approach is particularly convenient for pricing American and Bermudan options, 
which can be exercised before the expiration date.
Numerical computations of European and American options are presented, and compared with results obtained with finite differences methods and with the Black-Scholes
prices. 
\end{abstract}


\section{Introduction}
 
Since the early nineties, a lot of research has been done on the topic of pure jump Lévy processes 
to describe the dynamics of the asset returns. The main contributions are
\cite{BN98}, \cite{EbKe95}, \cite{GeMaYo98}, \cite{MaSe90}.

Lévy processes are stochastic processes with independent and stationary
increments that have nice analytical properties and 
reproduce quite well the statistical features of the financial data (see for instance \cite{Ait12} and \cite{Cont}).
In Figure \ref{FigPDF} we present as examples the histograms of the
daily log-returns for four of the major indices:  
the S\&P 500 Stock Index, 
the KOSPI (Korea Composite Stock Price Index), 
XAO (All Ordinaries Australian Index)  
and TAIEX (Taiwan Capitalization weighted Stock Index).
In the pictures we show the Normal and Variance Gamma (VG) densities fitted with the market data. 
It is straightforward to check that the 
VG density reproduces much better the high peaks near the origin, and the heavy tails feature.

The \emph{Variance Gamma} process is a pure jump Lévy process with infinite activity.
This means that when the magnitude of the jumps becomes infinitesimally small, the arrival rate of jumps tends to infinity.
The first complete presentation of financial applications of the symmetric VG model is given in \cite{MaSe90} 
where, with respect to the Gaussian model, an additional parameter is introduced in order to control the kurtosis 
(the skewness is not considered).
The authors model the log-returns with a driftless Brownian motion whose variance is Gamma distributed. 
This is the origin of the name ``Variance Gamma''.
\begin{figure}[t!]
 \centering
 \includegraphics[width=0.47\textwidth]{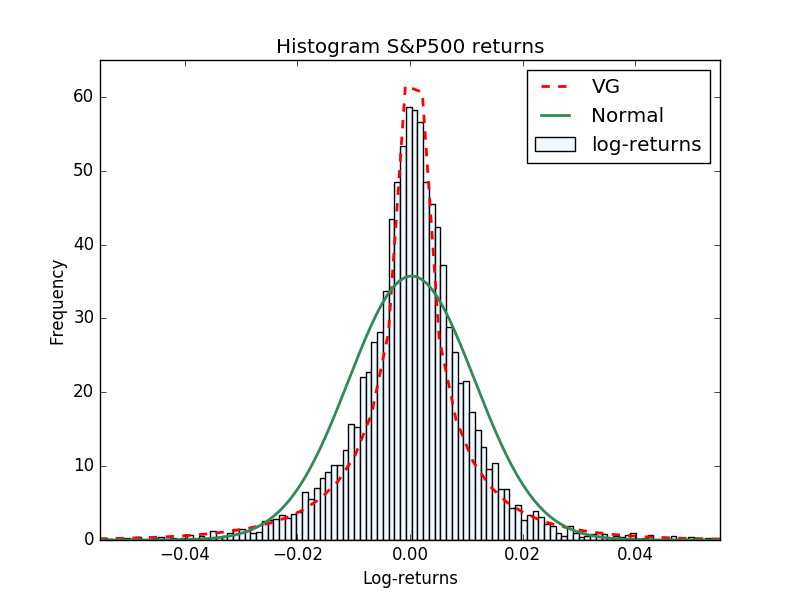}
 ~
 \includegraphics[width=0.47\textwidth]{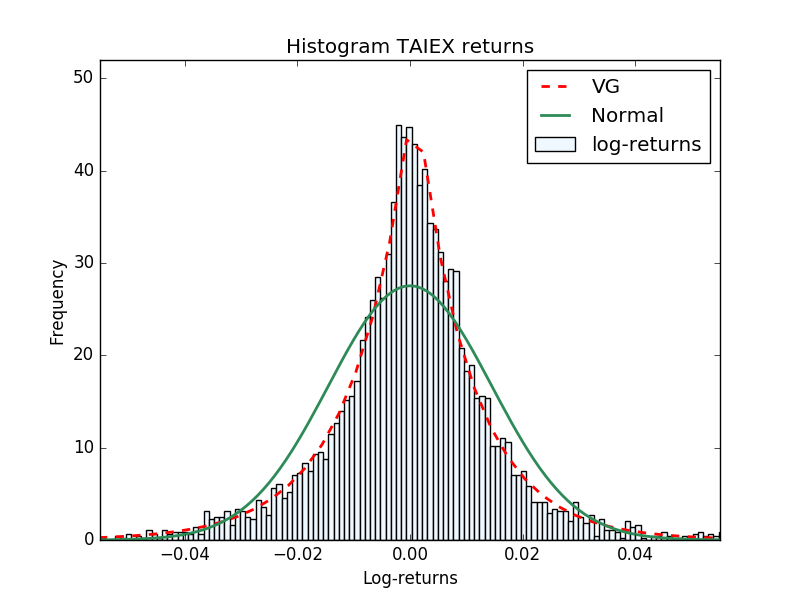}
 ~
 \includegraphics[width=0.47\textwidth]{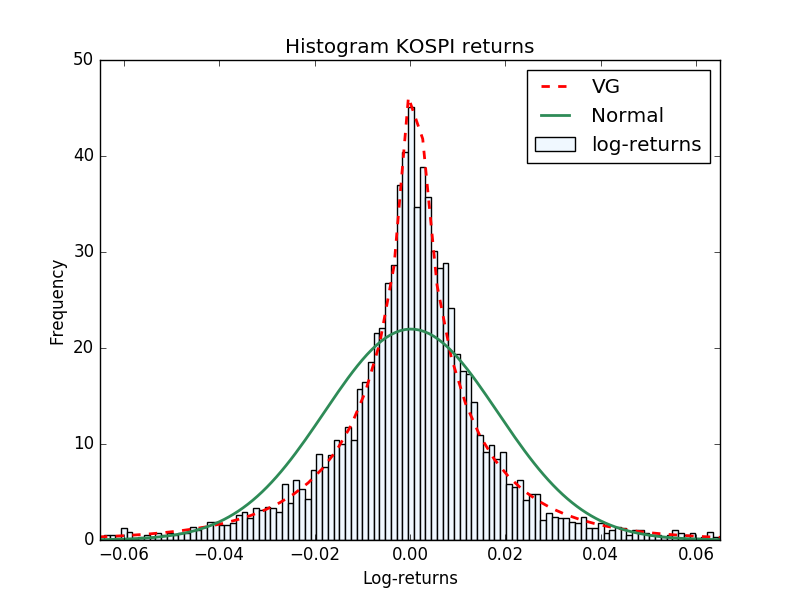}
 ~
 \includegraphics[width=0.47\textwidth]{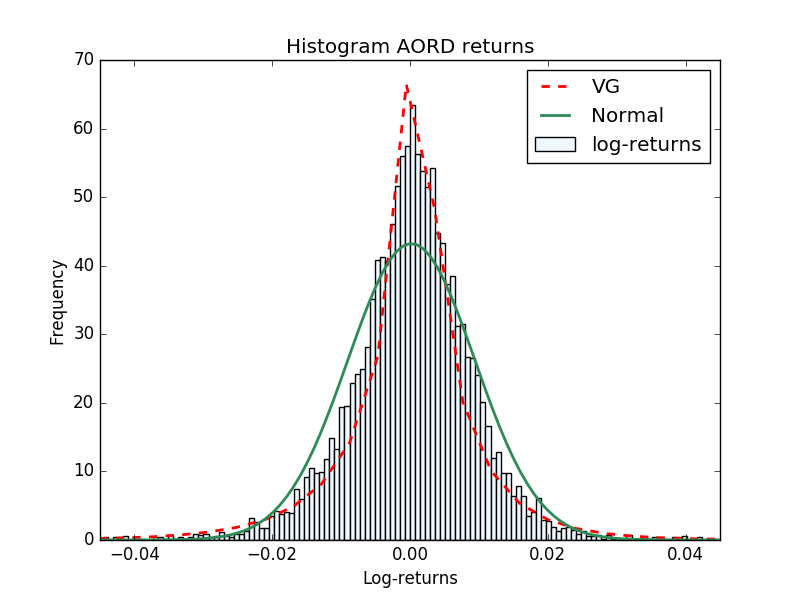}
 % VG4_pdf.png: 0x0 pixel, 0dpi, 0.00x0.00 cm, bb=
 \caption{Histograms of daily log-returns for S\&P500, KOSPI, XAO and TAIEX. The dashed line corresponds to the VG density (\ref{pdf_VG}). 
 The continuous line is the normal density. The parameters for both densities are obtained by the \emph{method of moments} (details in \cite{Se04}).}
 \label{FigPDF}
\end{figure} 

It is possible to give two representations to the VG process. In the first one,
the VG process is obtained by time changing a Brownian motion with drift:
the Brownian motion is evaluated at random times that are Gamma distributed.
The economic interpretation is that the trading relevant times are indeed random. 
The non-symmetric VG process is described in \cite{MCC98}, 
where the authors also present an explicit form of the density function, 
and a closed formula for the price of a vanilla European option.
The authors consider a Brownian motion with a non-zero drift, and this additional parameter allows to control the skewness as well.

The VG process 
has an infinite number of jumps in any time interval and, unlike the Brownian motion, it does not have a continuous martingale component.
Another important difference from the Brownian motion is that the VG process has finite variation, therefore the sum of the
absolute value of the increments in any time interval converges.
This fact can be derived easily from the second representation for VG processes: every VG can be represented as the difference of two 
(finite variation) Gamma processes.
The proof can be found in \cite{MCC98}, where the authors show that the two representations are equivalent, 
and also derive the VG characteristic function from the product
of two Gamma characteristic functions. 
The second representation has an interesting economical interpretation:
it can be seen as the difference of gains and losses. The Gamma
processes are always increasing, therefore this representation is coherent with independent gains and losses processes.

The VG process was first presented in the context of option pricing 
in \cite{MaMi91}, where it has been used for pricing European options.
European vanilla options can be easily priced by the analytical formula presented in \cite{MCC98} and exotics
can be priced numerically by several techniques:
Monte Carlo methods for VG are presented in \cite{Fu00}. 
A finite difference scheme for the VG Partial Integro-Differential Equation (PIDE) 
is described in \cite{CoVo05b}. In \cite{CaMa98}, the authors show how to price options by using a Fourier transform approach.
The problem for American options is considered in \cite{Al05}, \cite{Oo05} and \cite{HiMa01}, where the authors present different finite difference 
schemes to solve the American VG PIDE.

The tree method was first introduced by \cite{CRR79} for a market where the log-price can change only in two different ways: 
an upward jump, or a downward jump. For this reason the model is called \emph{binomial model}. The authors 
prove that when the number of time steps goes to infinity, the discrete random walk of the log-price converges to the Brownian motion
and the option price converges to the Black-Scholes price (\cite{BS73}).
The \emph{multinomial model} is a generalization of the binomial model, and at each time step it considers 
more than just two possible future states.
A general multinomial method for pricing European and American options under exponential L\'evy processes is described in \cite{MaSoSz}. 
In \cite{KeWe06} the authors consider a multinomial method for general exponential L\'evy processes based on the moment matching condition.
Other methods based on the moment matching condition are for instance \cite{HaMac10}, with applications to the Normal Inverse Gaussian process, and 
\cite{See13} with applications to the VG process.  
In the present work we consider a multinomial discretization based on the cumulant matching condition as explained in \cite{YaPr01}, \cite{YaPr03} and \cite{YaPr06}.

In Section \ref{sec2} we present the basic features of Lévy processes, in particular finite variation processes. 
The VG process and exponential VG are introduced in the successive subsections. 
A short summary of
some useful concepts such as integration with respect to the Poisson random measure, and the relation between the Lévy symbol and the cumulants are collected in 
the Appendices \ref{App1} and \ref{App2}. 
In Section \ref{sec3} we review the construction of the multinomial tree, following the method of
moment matching proposed in \cite{YaPr01}. We prove that the multinomial tree converges to the continuous time jump process that we have
introduced to approximate the VG process.
In Section \ref{sec4}, which is the most important of the paper, we describe the algorithm for pricing options with the 
multinomial method and show the numerical results for European and American options.
In Section \ref{sec5} we present a topic that deserves further research. 
We show how to obtain the parameters of the discrete time Markov chain that approximates the jump process, by discretizing its infinitesimal generator.
However, with this method the transition probabilities are not always positive. These probabilities, in general, are different from the probabilities obtained by
the moment matching condition, but for a particular choice of the parameters we argue that the two must coincide. This topic needs to be further investigated.
In Section \ref{sec6} we present the conclusions.

\section{Lévy processes}\label{sec2}

Let $L_t$ be a stochastic process defined on a probability space $(\Omega,\mathcal{F},(\mathcal{F}_{t \ge 0}),\PP)$,
$L_t$ is said to be a \emph{Lévy process} if it satisfies the three properties:
\begin{enumerate}
 \item $L_{0}=0$.
 \item $L_t$ has independent and stationary increments.
 \item $L_t$ is stochastically continuous: 
 $\forall \epsilon,t > 0 \; \; \lim_{h\to 0} \PP \bigl( |L_{t+h}-L_t|>\epsilon \bigr)=0.  $
\end{enumerate}
The characteristic function of a Lévy process $L_t$ has the Lévy-Khintchine representation:
\begin{align}\label{LevyKin}
\phi_{L_t}(u)  &= \mathbb{E} [e^{iu L_t}]  \\ 
	   &= e^{t \eta(u)} \nonumber \\
	   &= \exp \left[t \biggl( ibu - \frac{1}{2} \tilde \sigma^2 u^2 + \int_{\R} 
	   \bigl( e^{i u x} -1 -iux \mathbbm{1}_{\{|x|<1\}} \bigr) \nu(dx) \biggr) \right], \nonumber		      
\end{align}
where $\eta(u)$ is called \emph{Lévy symbol}, $b\in \R$ and $\tilde \sigma \geq 0$ are constants\footnote{The diffusion coefficient is usually called $\sigma$.
Here we use $\tilde \sigma$ because $\sigma$ will be used for the VG process.} 
and $\nu(dx)$ is the \emph{Lévy measure}
which satisfies:
\begin{equation}\label{Levy_m}
 \nu (\{ 0 \} ) = 0, \hspace{2em}
 \int_{\R} (1\wedge x^2) \nu(dx) < \infty.
\end{equation}
The \emph{Lévy triplet} $(b,\tilde \sigma,\nu)$ completely characterizes a Lévy process.
Every Lévy process can be written as the superposition of a drift, a Brownian motion and two pure jump processes.
This is the so called \emph{Lévy-It\={o} decomposition}: 
\begin{equation}\label{Levy_Ito}
  dL_t = bdt + \tilde \sigma dW_t + \int_{|x|\geq1} x N(dt,dx) + \int_{|x|<1} x \tilde{N}(dt,dx),
\end{equation}
where $W_t$ is a standard Brownian motion, 
$N(dt,dx)$ and $\tilde N(dt,dx)$ are the Poisson random measure and the compensated Poisson random measure (see Appendix \ref{App1}).

We are interested in particular in processes with finite variation and finite moments. 
We see that the Lévy measure contains all the
information we need:
\begin{itemize}
 \item A Lévy process with triplet $(b,\tilde \sigma,\nu)$ is of finite variation if and only if
 \begin{equation}\label{fin_variation} 
  \tilde \sigma = 0 \hspace{1 em}\mbox{and} \hspace{1 em} \int_{|x|<1} |x| \nu(dx) < \infty. 
 \end{equation} 
 \item A Lévy process has a finite moment of order $n$, $\E[X_t^n] < \infty$, if and only if
 \begin{equation}\label{moment_condition}
  \int_{|x|\geq 1} |x|^n \nu(dx) < \infty.
 \end{equation}
\end{itemize}
For a proof see \cite{Applebaum}, Theorem 2.4.25 and Theorem 2.5.2.
For processes with finite variation, the truncator term in (\ref{LevyKin}) can be absorbed in the parameter $b' = b - \int_{|x|<1} x \nu(dx)$.
It is easy to verify that every finite variation Lévy process can be represented as an integral of a Poisson random measure:
\begin{equation}\label{Levy_Ito_Poisson}
 L_t = b' t + \int_{\R \backslash \{0\}} x N(t,dx),
\end{equation}
The Lévy symbol is:
\begin{equation}\label{Character_Poisson}
 \eta(u) = ib'u + \int_{\R} (e^{iux} -1) \nu(dx). 
\end{equation}

\subsection{The Variance Gamma process}
The VG process is obtained by time changing a Brownian motion with drift. The new time variable is a stochastic process 
$T_t \sim \Gamma(\mu t,\kappa t)$ with density\footnote{Usually the Gamma distribution is 
parametrized by a shape and scale positive parameters $T \sim \Gamma(\rho,\zeta)$. The Gamma process 
$T_t \sim \Gamma(\rho t,\zeta)$ has pdf 
$f_{T_t}(x) = \frac{\zeta^{-\rho t}}{\Gamma(\rho t)}x^{\rho t -1}e^{-\frac{x}{\zeta}}$ and has moments $\E[T_t]=\rho \zeta t$ 
and $\mbox{Var}[T_t] = \rho \zeta^2 t$. Here we use a parametrization as in \cite{MCC98} such that $\E[T_t]=\mu t$ and $\mbox{Var}
[T_t] = \kappa t$, so $\zeta=\frac{\kappa}{\mu}$, $\rho=\frac{\mu^2}{\kappa}$.}:
\begin{equation}
 f_{T_t}(x)= \frac{(\frac{\mu}{\kappa})^{\frac{\mu^2 t}{\kappa}}}{\Gamma(\frac{\mu^2 t}{\kappa})}x^{\frac{\mu^2 t}{\kappa} -1}
 e^{-\frac{\mu x}{\kappa}} \hspace{2em} x \geq 0.
\end{equation}
The Gamma process $T_t$ is a \emph{subordinator}. A subordinator is a one dimensional Lévy process that is 
non-decreasing almost surely. Therefore it is suitable to represent a time variable. 
It is possible to prove that every subordinator is a finite variation process (see \cite{Applebaum}).  

Considering a Brownian motion with drift $X_t = \theta t + \sigma W_t$, with $W_t \sim \mathcal{N}(0,t)$, let's replace the time variable by the Gamma 
subordinator $T_t \sim \Gamma(t,\kappa t)$ (with $\mu=1$).
We obtain the \emph{Variance Gamma} process:
\begin{equation}\label{VG_process}
 X_{t} = \theta T_t + \sigma W_{T_t} .
\end{equation}
It depends on three parameters:
\begin{itemize}
 \item $\theta$, the drift of the Brownian motion,	
 \item $\sigma$, the volatility of the Brownian motion,
 \item $\kappa$, the variance of the Gamma process.
\end{itemize}
The probability density function of $X_t$ can be computed conditioning on the realization of $T_t$:
\begin{align}\label{pdf_VG}
 f_{X_t}(x) &= \int f_{X_t,T_t}(x,y) dy = \int f_{X_t|T_t}(x|y) f_{T_t}(y) dy \\ \nonumber
         &= \int_0^{\infty} \frac{1}{\sigma \sqrt{2\pi y}} e^{-\frac{(x -\theta y)^2}{2\sigma^2 y}}
         \frac{y^{\frac{t}{\kappa} -1}}{\kappa^{\frac{t}{\kappa}} \Gamma(\frac{t}{\kappa})}
          e^{-\frac{y}{\kappa}} \, dy \\ \nonumber
         &= \frac{2 \exp(\frac{\theta x}{\sigma^2})}{\kappa^{\frac{t}{\kappa}} \sqrt{2\pi}\sigma \Gamma(\frac{t}{\kappa}) }
            \biggl( \frac{x^2}{2\frac{\sigma^2}{\kappa} + \theta^2} \biggr)^{\frac{t}{2\kappa}-\frac{1}{4}} 
            K_{\frac{t}{\kappa}-\frac{1}{2}} 
            \biggl( \frac{1}{\sigma^2} \sqrt{x^2 \bigl(\frac{2\sigma^2}{\kappa}+\theta^2 \bigr)} \biggr),
\end{align}
where the function $K$ is a modified Bessel function of the second kind (see \cite{MCC98} for the computations).
The characteristic function can be obtained by conditioning too: 
\begin{align}
 \phi_{X_t}(u) &= \biggl( 1-i \kappa \bigl( u\theta +\frac{i}{2}\sigma^2 u^2 \bigr) \biggr)^{-\frac{t}{\kappa}} \\  
	       &= \biggl( 1-i\theta \kappa u + \frac{1}{2} \sigma^2 \kappa u^2 \biggr)^{-\frac{t}{\kappa}},
\end{align}
see proposition 1.3.27 in \cite{Applebaum}. 
The Lévy symbol is:
\begin{equation}\label{Levy_exp_VG}
 \eta(u) = -\frac{1}{\kappa} \log(1-i\theta \kappa u + \frac{1}{2} \sigma^2 \kappa u^2).
\end{equation}
Using the formula (\ref{cumulants}) in Appendix \ref{App2} for the cumulants we derive:
\begin{align}\label{VG_cumulants}
 c_1 &= t\theta \\ \nonumber
 c_2 &= t(\sigma^2 + \theta^2 \kappa) \\ \nonumber
 c_3 &= t (2\theta^3\kappa^2 + 3 \sigma^2 \theta \kappa) \\ \nonumber
 c_4 &= t (3\sigma^4 \kappa + 12\sigma^2 \theta^2 \kappa^2 +6\theta^4\kappa^3)\nonumber . 
\end{align}
The VG Lévy measure is\footnote{In \cite{MCC98} the authors derive the expression for the Lévy measure 
by representing the VG process as the difference of two Gamma processes.}:
\begin{equation}\label{VG_measure}
 \nu(dx) = \frac{e^{\frac{\theta x}{\sigma^2}}}{\kappa|x|} \exp 
 \left( - \frac{\sqrt{\frac{2}{\kappa} + \frac{\theta^2}{\sigma^2}}}{\sigma} |x|\right) dx,
\end{equation}
and it satisfies conditions (\ref{fin_variation}) and (\ref{moment_condition}). 
The VG process can be represented as a 
pure jump process as in (\ref{Levy_Ito_Poisson}) and (\ref{Character_Poisson}), with no additional drift $b'=0$.
\begin{equation}\label{Poisson_VG}
 X_t = \int_{\R \backslash \{0\}} x N(t,dx).
\end{equation}
All the informations are contained in the Lévy measure (\ref{VG_measure}),
which completely describes the process. Even if the process has been created by Brownian
subordination, it has no diffusion component.  
The L\'evy triplet is $( \int_{|x|<1} x \nu(dx), 0, \nu)$.
Using the formalism of Poisson integrals in Appendix \ref{App1}, the Lévy symbol (\ref{Levy_exp_VG}) 
has the representation\footnote{See Example 8.10 in \cite{Sato}.}:
\begin{equation}
 \eta(u) =\int_{\R} (e^{iux} -1) \nu(dx). 
\end{equation}

\subsection{Exponential VG model}

Under the risk neutral measure $\Q$, the dynamics of the stock price is described by an \emph{exponential Lévy model}:
\begin{equation}
 S_t = S_0 e^{rt + L_t},
\end{equation}
where $r$ is the risk free interest rate, and $L_t$ is a general Lévy process.
Under $\Q$, the discounted price is a $\Q$-martingale:
\begin{equation}
 \E^{\Q} \bigl[ S_te^{-rt} \bigr| S_0 \bigr] =  \E^{\Q} \bigl[ S_0e^{L_t} \bigr| S_0 \bigr] = S_0, 
\end{equation}
and thus $\E^{\Q}[ e^{L_t} | L_0=0] = 1 $.
The condition for the existence of the exponential moment $\E[ e^{L_t}] < \infty$ is equivalent to:  
\begin{equation}
 \int_{|x|>1} e^x \nu(dx) < \infty,
\end{equation}
as proved in Lemma 25.7 in \cite{Sato}. The VG process $X_t$ has finite exponential moment.
In order to satisfy the martingale condition\footnote{
To obtain the correction term $\omega$ we have to find the exponential moment of $X_t$ using its characteristic function: 
\begin{align*}
\E[ e^{X_t}] = \phi_{X_t}(-i) = e^{-\omega t}
 \end{align*}.} we need to add a correction term to $X_t$. 
The following process is a martingale:
\begin{equation}\label{S_pprocess}
 e^{-rt}S_t = S_0 e^{\omega t + X_t}.
\end{equation}
where $w = \frac{1}{\kappa} \log(1-\theta \kappa -\frac{1}{2}\sigma^2 \kappa)$.
Passing to the log-price $Y_t = \log(S_t)$, we get a process in the form of Eq. (\ref{Levy_Ito_Poisson}) with $b' = r + \omega$:
\begin{equation}\label{log_process}
 Y_t = Y_0 + (r+\omega)t + \int_{\R \backslash \{0\}} x N(t,dx).
\end{equation}
Let $V(t,Y_t)$ be the value of an option at time $t$. 
We assume that $V(t,y) \in C^{1,1}([t_0,T],\R)$ and has a polynomial growth at infinity.\\ 
By the \emph{martingale pricing theory}, the discounted price of the option
is a martingale and it is possible to derive the PIDE for the price of the option:
\begin{equation}
 \E^{\Q}\biggl[ d\bigl( e^{-rt}V(t,Y_t) \bigr)\biggr] = \frac{\partial V(t,y)}{\partial t} + \LL^{Y_t} V(t,y) -r V(t,y) = 0,   
\end{equation}
where $\LL^{Y_t}$ is the infinitesimal generator of the log-price process (\ref{log_process}). 
The resulting PIDE is:
\begin{equation} \label{VG_PIDE}
 \frac{\partial V(t,y)}{\partial t} + (r+\omega) \frac{\partial V(t,y)}{\partial y}
 + \int_{\R \backslash \{0\}} \bigl[ V(t,y+x) - V(t,y) \bigr] \nu(dx) = rV(t,y) .
\end{equation}

\section{The multinomial method} \label{sec3}
In this section we introduce the multinomial method proposed in \cite{YaPr06}. 
The stock price is considered as a Markov chain with $L$ possible future states at each time. 
In this setting, the time $t \in [t_0,T]$ is discretized as $t_n = t_0 + n\Delta t$ for $n=0, ... ,N$ and 
$\Delta t = (T-t_0)/N$. We denote the stock price at time $t_n$ as $S(t_n) = S_n$.

Consider the up/down factors $u>d>0$ and write the discrete evolution of the stock price $S_n$ as:
\begin{equation}\label{Discr_S}
  S_{n+1} = u^{L-l}d^{l-1} S_n  \hspace{3em} l=1, ... , L  
\end{equation}
where each future state has transition probability $p_l$, satisfying $\sum_{l=1}^L p_l = 1$.
The value of the stock at time $t_n$ can assume $j \in [1,...,n(L-1)+1]$ possible values:
\begin{equation}\label{Discr_S2}
  S_{n}^{(j)} = u^{n(L-l)+1-j}d^{j-1} S_0.  
\end{equation}
The multinomial tree is recombining if for a constant $c>1$, $u/d = c$.
Regarding the present work, we only consider five branches, $L=5$. As explained in \cite{YaPr06}, this number of branches is 
enough to model the features of a stochastic process up to its fourth moment.

\subsection{Moment matching}

To determine the parameters of the Markov chain we require that its local moments are equal to that of the continuous process.
First, we rewrite the continuous process (\ref{log_process}) as the sum of a drift term and a martingale term:
\begin{align}\label{log_martingale}
 Y_{t+\Delta t}-Y_t &= (r+\omega)\Delta t + \int_{\R \backslash \{0\}} x N(\Delta t,dx) \\ \nonumber
		    &= (r+\omega + \theta)\Delta t + \int_{\R \backslash \{0\}} x \tilde N(\Delta t,dx)
\end{align}
where $\theta = \int_{\R \backslash \{0\}} x \nu(dx) = \E \bigl[ \int_{\R \backslash \{0\}} x N(1,dx) \bigr]$ is the expected value of the VG process in \ref{Poisson_VG}, 
when $\Delta t=1$. 
The integral with respect to the compensated Poisson measure $\tilde N(\Delta t,dx)$ is a martingale (see Appendix \ref{App1}).

We can pass to log-prices $Y_n = \log(S_n)$ in the discrete Eq. (\ref{Discr_S}), and write it as the sum of a drift component and a 
random variable with $L$ possible outcomes:
\begin{align}\label{Discr_Y}
 \Delta Y = Y_{n+1} - Y_n &= (L-l) \log(u) + (l-1) \log(d) \\ \nonumber
  &= \bar b\, \Delta t + (L-2l+1) \alpha(\Delta t).
\end{align}
The term $\bar b\, \Delta t$ is the drift term, while $l$ is a random variable that assumes values in $[1,2,...,L]$ with probability $p_l$. 
It has to satisfy the martingale condition:
$$ \E \bigl[(L-2l+1) \alpha(\Delta t) \bigr] = \alpha(\Delta t) \sum_{l=1}^L p_l (L-2l+1) = 0, $$
with $\alpha(\Delta t)$ a function of $\Delta t$.

The corresponding up/down factors have the following representation:
\begin{equation}\label{updown}
 u = \exp\biggl( \frac{b}{L-1} + \alpha(\Delta t) \biggr) \hspace{2em}  d = \exp\biggl( \frac{b}{L-1} - \alpha(\Delta t) \biggr),
\end{equation}
and we can readily see that if $u/d$ is constant the tree recombines.

Given the mean $c_1 = \E[\Delta Y] = \bar b \Delta t$, the $k$-central moment is:
\begin{equation}\label{higher_moments}
 \E \bigl[ (\Delta Y - c_1)^k \bigr] = \alpha(\Delta t)^k\, \E \bigl[ (L-2l+1)^{k} \bigr].
\end{equation}
The moment matching condition requires that the central moments of the discrete process (\ref{Discr_Y}) 
are equal to the central moments
of the continuous process (\ref{log_martingale}):
\begin{equation}\label{moment_matching}
 \alpha(\Delta t)^k\, \E \bigl[ (L-2l+1)^{k} \bigr] = \mu_k.
\end{equation}
We fix $L=5$, and using the relation between central moments and cumulants (Eq. (\ref{moment_cumulants}) in Appendix \ref{App2}) 
we solve the linear system of equations for the transition probabilities:
\begin{align}\label{probabilities1}
 p_1 &= \frac{1}{196 \alpha(t)^4} \biggl[ \frac{3}{2} c_2^2 -2 c_2\alpha(t)^2 + 2 c_3 \alpha(t) +\frac{1}{2} c_4  \biggr] \\ \nonumber
 p_2 &= \frac{1}{196 \alpha(t)^4} \biggl[ -6 c_2 + 32c_2 \alpha(t)^2 - 4c_3 \alpha(t) -2 c_4 \biggr] \\ \nonumber
 p_3 &=  1 + \frac{1}{196 \alpha(t)^4} \biggl[ 3c_4 + 9c_2^2 -60c_2 \alpha(t)^2   \biggr] \\ \nonumber
 p_4 &= \frac{1}{196 \alpha(t)^4} \biggl[ -6 c_2 + 32c_2 \alpha(t)^2 + 4c_3 \alpha(t) -2 c_4  \biggr] \\ \nonumber
 p_5 &= \frac{1}{196 \alpha(t)^4} \biggl[ \frac{3}{2} c_2^2 -2 c_2\alpha(t)^2 - 2 c_3 \alpha(t) +\frac{1}{2} c_4 \biggr].
\end{align}
The drift parameter is $\bar b = r+\omega + \theta$.
The only missing parameter to determine is $\alpha(\Delta t)$. This is a function of the time increment $\Delta t$ and can be determined using the 
higher order terms in the moment matching condition together with the condition of positive probabilities.

Recall that the well known binomial model \cite{CRR79} assumes the value $\alpha(\Delta t) = \tilde \sigma \sqrt{\Delta t}$,
that represents the volatility of the increments in the time interval $\Delta t$.
In the trinomial model, the parameter $\alpha(\Delta t)$ assumes value $\alpha(\Delta t) = \frac{1}{2} \tilde \sigma \sqrt{3\Delta t}$, see for instance \cite{YaPr01}.
For the multinomial method a good representation for $\alpha(\Delta t)$ is:
\begin{equation}\label{alphat}
 \alpha(\Delta t) = \sqrt{c_2} \sqrt{\frac{3+\bar \kappa}{12}},
\end{equation}
where $\bar \kappa = c_4 / c_2^2$ is the excess of kurtosis\footnote{We use the bar over $\kappa$, 
to distinguish the kurtosis from the variance of the gamma process $\kappa$.}. 
We refer to \cite{YaPr06} for the derivation.
This choice guarantees that the probabilities $p_i$ for $i=1...5$ are always positive and sum to one. We can replace the expression
(\ref{alphat}) inside (\ref{probabilities1}), to obtain the simpler form:
\begin{align}\label{probabilities2}
 & [p_1,p_2,p_3,p_4,p_5] \approx \biggl[ \frac{3+\bar \kappa+s\sqrt{9+3\bar \kappa}}{4(3+\bar \kappa)^2} , 
 \frac{3+\bar \kappa-s\sqrt{9+3\bar \kappa}}{2(3+\bar \kappa)^2} , \\ \nonumber
 &
 \frac{3+2\bar \kappa}{2(3+\bar \kappa)} ,
 \frac{3+\bar \kappa+s\sqrt{9+3\bar \kappa}}{2(3+\bar \kappa)^2} ,
 \frac{3+\bar \kappa-s\sqrt{9+3\bar \kappa}}{4(3+\bar \kappa)^2} \biggr],
\end{align}
where $s = c_3 / \sqrt{c_2^3}$ is the skewness.
\begin{remark}
 The standard deviation of every L\'{e}vy process with finite moments follows the square root rule. This means that the term $\alpha(\Delta t)$ has to be proportional
 to the square root of $\Delta t$. In the binomial and trinomial models, the proportionality constant is explicit, while for the pentanomial method it is implicit
 in the formula (\ref{alphat}). Expanding the formula using the expression (\ref{VG_cumulants}) for the cumulants, it is possible to check that the square root rule is
 satisfied at first order in $\sqrt{\Delta t}$.
\end{remark}

\subsection{Convergence}

We call a generic jump process (\ref{Levy_Ito_Poisson}) with first four cumulants $c_1$,$c_2$,$c_3$,$c_4$ equal to the VG cumulants (\ref{VG_cumulants}), 
the \emph{approximated process} $X^A$. 
The cumulant generating function of the increment $\Delta X^A$ has the following series representation (see Appendix (\ref{App2})):
\begin{equation}\label{cum_gen_appr}
 H_{\Delta X^{A}}(u) = ic_1 u -\frac{c_2u^2}{2} -\frac{ic_3u^3}{3!} + \frac{c_4u^4}{4!} + \mathcal{O}(u^5).
\end{equation}
We are interested in the approximation of a VG process with drift (\ref{log_martingale}), therefore we require that $c_1 = \bar b \Delta t = (r+\omega+\theta)\Delta t$. 
\begin{theorem}
The increments of the discrete Markov chain (\ref{Discr_Y}) and the increments of the approximated process $X^A$ have the same distribution by construction.
\end{theorem}
\begin{proof}
The idea of the proof is to show that the cumulant generating function of the discrete process (\ref{Discr_Y}) 
coincides with that of the approximated process (\ref{cum_gen_appr}). We prove it using the moment
matching condition (\ref{moment_matching}).
\begin{align}
H_{\Delta Y}(u) &= \log \bigl( \phi_{\Delta Y}(u)  \bigr) = \log \biggl( \E \bigl[ e^{iu \Delta Y} \bigr] \biggr) \\ \nonumber
                &= \log \biggl( \E \biggl[ e^{iu \bigl( \bar b \Delta t + (L-2l+1) \alpha(\Delta t) \bigr) } \biggr] \biggr) \\ \nonumber
		&= iu \bar b \Delta t + \log \biggl( \E \biggl[ e^{iu \bigl(  (L-2l+1) \alpha(\Delta t) \bigr) } \biggr] \biggr).
\end{align}
We can expand the exponential function in Taylor series and use the moment matching condition (\ref{moment_matching}) to obtain:
\begin{align}
H_{\Delta Y}(u) &= iu \bar b \Delta t + \log \biggl( \sum_{k=0}^{\infty} \frac{(iu)^k}{k!} 
\bigl(\alpha(\Delta t)\bigr)^k \E \biggl[ \bigl( L-2l+1  \bigr)^k \biggr] \biggr) \\ \nonumber
                &= iu \bar b \Delta t + \log \biggl( \sum_{k=0}^{\infty} \frac{(iu)^k}{k!} \mu_k \biggr) \\ \nonumber
                &= iu c_1 + \sum_{k=0}^{\infty} \frac{(iu)^k}{k!} c_k \\ \nonumber 
                &= H_{\Delta X^{A}}(u),
\end{align}
\end{proof}
\begin{remark}
All the cumulants of $\Delta X^A$ are equal to the cumulants of the Markov chain (\ref{Discr_Y}) by construction, but only the first four are equal to the VG cumulants.
When all the cumulants $c_i$, for $0 \leq i \leq n$, are equal to the VG cumulants, the approximated process $X^A$ converges to
the original VG process for $n \to \infty$.
To control $n$ cumulants, we need $n+1$ branches. Therefore, when the number of cumulants of $\Delta X^A$ equal to those of the VG process goes to infinity, 
the number of branches have to go to infinity as well.
We assume that five branches ($L=5$) are enough to describe the features of the underlying process and, at the same time, keep the numerical
problem simple. 
\end{remark}

\begin{theorem}
The distribution of the pentanomial tree at time $N$ converges to the distribution of a compound Poisson process at time $N$ with $L=5$ possible jump sizes and activity $\lambda = \frac{3}{2 \bar \kappa N}$, when $\Delta t \to 0$.   
\end{theorem}
For the proof of this theorem  
we refer to Section 4.2 of \cite{YaPr06}. The authors first define the jump sizes and their respective probabilities, and then
prove that when $\Delta t \to 0$ the characteristic function of the pentanomial tree converges to the 
characteristic function of the compound Poisson process.

\section{Numerical results} \label{sec4}

In this section we present the steps to implement the algorithm for pricing European and American options with the multinomial method.
Then we compare the results with those obtained by the PIDE method and Black-Scholes model.

\subsection{Algorithm}

We suggest the following algorithm for pricing with the multinomial method:
\begin{enumerate}
 \item Compute the transition probabilities vector (\ref{probabilities2}). 
 \item Compute the up/down factors $u$ and $d$ (\ref{updown}) and the vector of prices $S_N$ at terminal time $N$ as in Eq. (\ref{Discr_S2}).
 \item Evaluate the payoff of the option $V^N(S_N)$ at terminal time $N$.
 \item Given the option's values at time $t_{n+1}$ compute the values at time $t_n$. The value is the conditional expectation:
 \begin{equation}
 V^n(s^{(k)}_n) = e^{-r\Delta t} \E^{\Q} \biggl[ V^{n+1}(S_{n+1}) \bigg| S^{(k)}_n = s^{(k)}_n \biggr]. 
\end{equation}
 \item If computing the price of an American option, the value at the previous time level is the maximum between the conditional expectation and
 the intrinsic value of the option. For an American put we have:
 \begin{equation}
 V^n(s^{(k)}_n) = \max \biggr \{ e^{-r\Delta t} \E^{\Q} \biggl[ V^{n+1}(S_{n+1}) \bigg| S^{(k)}_n = s^{(k)}_n \biggr] , K-s^{(k)}_n \biggr \}. 
\end{equation}	
 \item Iterate the algorithm until the initial time $t_0$. 
\end{enumerate}
In the parameters calibration procedure, sometimes it is common to estimate first the historical parameters, and use them as initial guess for the
least squares minimization that recovers the risk neutral parameters.
In \cite{Se04} are presented several methods for historical parameters estimation of the VG density. We use the simple method of moments to estimate the 
parameters for the data in Fig. \ref{FigPDF}.
In all future calculations we consider the risk neutral VG parameters in Table \ref{sample-table}.

\begin{table}[!h]
\tbl{$r$ is the risk free interest rate and $\theta, \sigma, \kappa$ are the risk neutral VG parameters.}
{\begin{tabular}{lccc} \toprule
%& \multicolumn{2}{l}{} \\ \cmidrule{2-4}
$r$ & $\theta$ & $\sigma$ & $\kappa$ \\ \midrule
0.06 & -0.1 & 0.2 & 0.2 \\ \bottomrule
\end{tabular}}
\label{sample-table}
\end{table}

\subsection{European options}

We compare the numerical results for European call and put options obtained with the multinomial and the PIDE approaches.
\begin{itemize}
 \item \emph{VG PIDE}: We solve the VG PIDE following the method proposed by \cite{CoVo05b}. The L\'{e}vy measure is singular in the origin
 and this is a problem for the computation of the integral term in (\ref{VG_PIDE}). 
 Fixing a truncation parameter $\epsilon >0$, we approximate the infinite activity martingale jump component with sizes smaller than $\epsilon$, with a Brownian motion with same variance.
 The resulting approximated PIDE has the ``jump-diffusion-like'' form:
 \begin{align}\label{VG_JD}
&  \frac{\partial V(t,x)}{\partial t} +
 \bigl( r-\frac{1}{2}\sigma_{\epsilon}^2 - w_{\epsilon} \bigr) \frac{\partial V(t,x)}{\partial x} 
 + \frac{1}{2}\sigma_{\epsilon}^2 \frac{\partial^2 V(t,x)}{\partial x^2} \\ \nonumber
 &+ \int_{|z| \geq \epsilon} V(t,x+z) \nu(dz) = (\lambda_{\epsilon} + r) V(t,x).
\end{align}
 where we introduced the parameters $\sigma_{\epsilon}^2 = \int_{|z| < \epsilon} z^2 \nu(dz)$, $\omega_{\epsilon} = \int_{|z| \geq \epsilon} (e^z-1) \nu(dz)$ and
 $\lambda_{\epsilon} = \int_{|z| \geq \epsilon} \nu(dz)$. More details are in \cite{Cont}.
 We solve the PIDE (\ref{VG_JD}) using the implicit-explicit finite difference scheme proposed in \cite{CoVo05b}, and choosing the truncator parameter $\epsilon = 1.5 \delta x$, 
 where $\delta x$ is the size of 
 the space step. It turns out that the solution of the discretized equation convergences very slowly to the option price, and therefore we required a grid with $14000$ space steps 
 and $7000$ time steps. The algorithm is written in Matlab and runs
 on an Intel i7 (7th Gen) with Linux. It takes about 30 minutes to complete. 
 \item \emph{Multinomial}: We follow the algorithm proposed in the previous section. The number of time steps for all the computations is $N=2000$. In the table \ref{Convergence} 
 we show a convergence table for the prices of European calls, puts and American puts. It is straightforward to see that the convergence is quite fast.  
\end{itemize}
\begin{table}[!h]
\tbl{Convergence table for ATM European and American options with strike $K=40$ and $T=1$. The time unit is in seconds.}
{\begin{tabular}{l|ccc|cc} \toprule
  $N$ & Eu. Call & Eu. Put & Time & Am. Put & Time \\
  \hline
    50 & 4.41873125 & 2.08928091 & 0.001 & 2.36765911 & 0.007 \\
    100 & 4.41960265 & 2.09015381 & 0.002 & 2.37255454 & 0.02 \\
    200 & 4.41997010 & 2.09052201 & 0.004 & 2.37480218 & 0.07 \\
    400 & 4.42013640 & 2.09068869 & 0.01 & 2.37587117 & 0.29 \\
    800 & 4.42021515 & 2.09076762 & 0.03 & 2.37639131 & 1.09 \\
    1000 & 4.42023054 & 2.09078306 & 0.04 & 2.37649417 & 1.67 \\
    1500 & 4.42025089 & 2.09080345 & 0.06 & 2.37663070 & 3.79 \\
    2000 & 4.42026098 & 2.09081357 & 0.10 & 2.37669869 & 6.80 \\
    2500 & 4.42026701 & 2.09081962 & 0.16 & 2.37673941 & 10.65 \\
    3000 & 4.42027102 & 2.09082364 & 0.2 & 2.37676652 & 14.78 \\
  \hline
\end{tabular}}
\label{Convergence}
\end{table}
Figures (\ref{figCall}) and (\ref{figPut}) show the prices obtained by the multinomial method compared with those obtained by PIDE.
In table \ref{Option_values} we compare directly the call/put numerical values obtained with the two methods.
\begin{figure}[t!]
 \centering
 \includegraphics[scale=0.5]{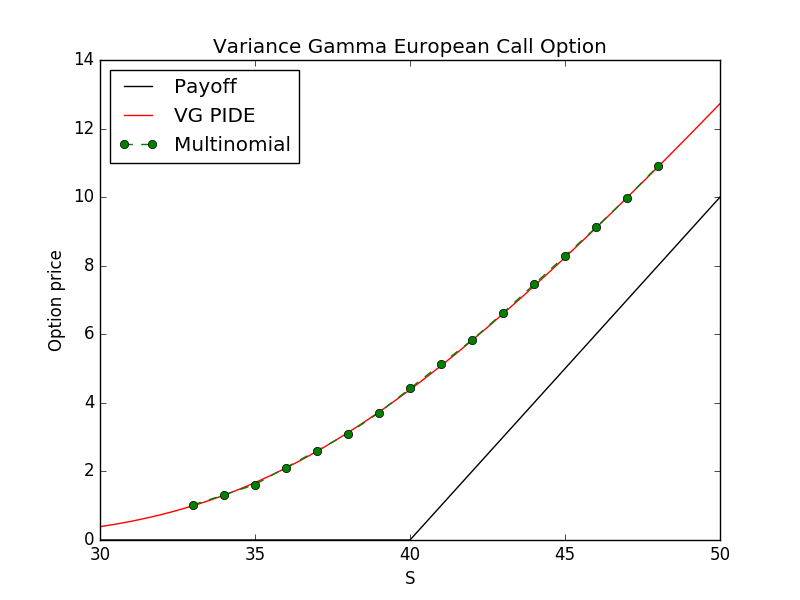}
 % EU_call.png: 0x0 pixel, 300dpi, 0.00x0.00 cm, bb=
 \caption{European call option with strike $K=40$ and time to maturity 1 year.}
 \label{figCall}
\end{figure}
\begin{figure}[t!]
 \centering
 \includegraphics[scale=0.5]{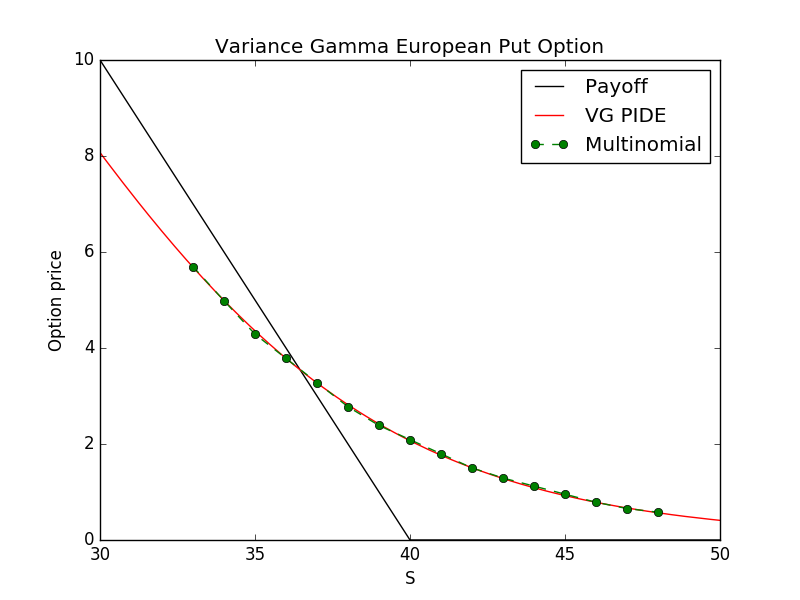}
 % EU_call.png: 0x0 pixel, 300dpi, 0.00x0.00 cm, bb=
 \caption{European put option with strike $K=40$ and time to maturity 1 year.}
 \label{figPut}
\end{figure}

Pricing vanilla call and put European options is quite simple and the best approach is to use the closed formula derived in \cite{MCC98}.
The big advantage of the multinomial method is in the computation of American options prices, where there is no closed formula and all the other approaches, 
such as PIDEs and Least Squares Monte Carlo, are difficult to implement and much slower.
\begin{table}[h!]
\tbl{European Options, with strike $K=40$ and $T=1$.}
{\begin{tabular}{lcccc} \toprule
& \multicolumn{2}{l}{Different methods} \\ \cmidrule{2-5}
  $S_0$ & PIDE Call & Multi Call & PIDE Put & Multi Put \\ \midrule
  36 & 2.1036 & 2.1131 & 3.7842 & 3.7837  \\
  38 & 3.1163 & 3.1051 & 2.7893 & 2.7756 \\
  40 & 4.4162 & 4.4203 & 2.0852 & 2.0908 \\
  42 & 5.8335 & 5.8309 & 1.5050 & 1.5014 \\
  44 & 7.4417 & 7.4524 & 1.1132 & 1.1229 \\ \bottomrule
\end{tabular}}
\label{Option_values}
\end{table}

\subsection{American options}

In this section we present the numerical results obtained with the multinomial method algorithm for American put options, and compare them with the PIDE method (see fig. \ref{AmVG}).
The PIDE (\ref{VG_JD}) is modified in order to take in account the early exercise feature:
\begin{align}\label{VG_Am_JD}
&  \min \biggl\{ - \frac{\partial V(t,x)}{\partial t} -
 \bigl( r-\frac{1}{2}\sigma_{\epsilon}^2 - w_{\epsilon} \bigr) \frac{\partial V(t,x)}{\partial x} 
 - \frac{1}{2}\sigma_{\epsilon}^2 \frac{\partial^2 V(t,x)}{\partial x^2} + (\lambda_{\epsilon} + r) V(t,x) \\ \nonumber
 &- \int_{|z| \geq \epsilon} V(t,x+z) \nu(dz) \, , \, \biggl( V(t,x) - (K-e^x)^+ \biggr) \biggr\} = 0.
\end{align}
To solve this equation we use the same settings used for the Eq. (\ref{VG_JD}): an implicit-explicit finite difference scheme, with 14000 space steps and 7000 time steps. 
The algorithm takes about 33 minutes to run.

The numerical values of the prices obtained with the multinomial and PIDE methods are collected in Tab. \ref{Option_values3}.
The run times for the multinomial algorithm are shown in the convergence table \ref{Convergence}.
\begin{figure}[ht!]
 \centering
 \includegraphics[scale=0.5]{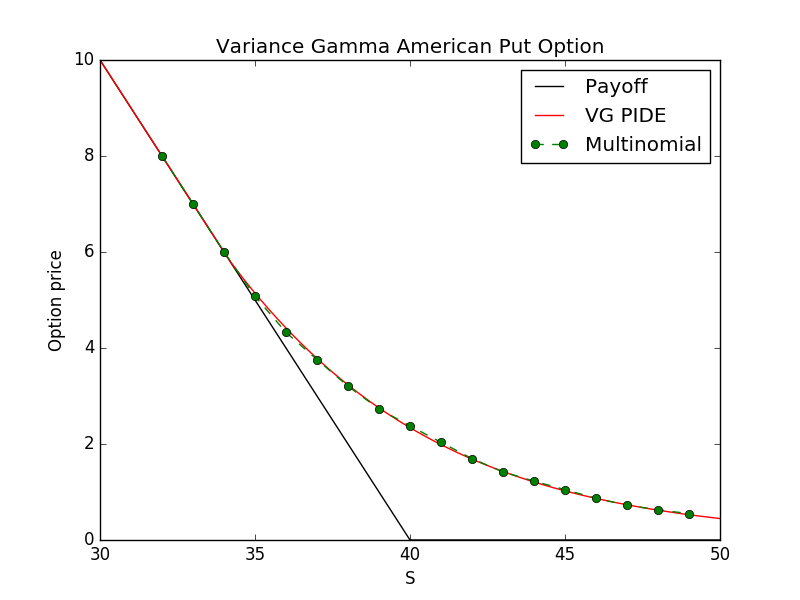}
 % EU_call.png: 0x0 pixel, 300dpi, 0.00x0.00 cm, bb=
 \caption{American put option with strike $K=40$ and time to maturity 1 year. Comparison of PIDE prices and multinomial prices.}
 \label{AmVG}
\end{figure}
\begin{table}[!h]
\tbl{Values for European and American put options using Black-Scholes and Variance Gamma model. Strike $K=40$ and $T=1$. The BS volatility have same value of the VG volatility:  
$ \sigma^{BS} = (\sigma^2 + \theta^2 \kappa) = 0.2049$. }
{\begin{tabular}{lrrrrr} \toprule
& \multicolumn{4}{r}{Prices comparison} \\ \cmidrule{2-6}
  $S_0$ & BS Eu. Put & VG Eu. Put & BS Am. Put & VG Am. Put & VG Am. PIDE Put \\ \midrule
  30 & 8.1316 & 8.0809 & 10     & 10     & 10 \\
  32 & 6.5292 & 6.4055 & 8      & 8      & 8 \\
  34 & 5.1169 & 4.9851 & 6.0894 & 6      & 6 \\
  36 & 3.9150 & 3.7837 & 4.5415 & 4.3173 & 4.3982 \\
  38 & 2.9263 & 2.7756 & 3.3264 & 3.2034 & 3.2195 \\
  40 & 2.1388 & 2.0908 & 2.3924 & 2.3767 & 2.3531 \\
  42 & 1.5322 & 1.5014 & 1.6911 & 1.6947 & 1.6849 \\
  44 & 1.0766 & 1.1229 & 1.1755 & 1.2267 & 1.2118 \\
  46 & 0.7433 & 0.7858 & 0.8043 & 0.8699 & 0.8650 \\
  48 & 0.5049 & 0.5787 & 0.5425 & 0.6310 & 0.6221 \\
  50 & 0.3384 & 0.4259 & 0.3612 & 0.4661 & 0.4480 \\ 
  52 & 0.2238 & 0.3015 & 0.2376 & 0.3242 & 0.3222 \\
  55 & 0.1178 & 0.1909 & 0.1243 & 0.2051 & 0.1990 \\
  60 & 0.0386 & 0.0880 & 0.0404 & 0.0942 & 0.0913 \\ \bottomrule
  \end{tabular}}
  \label{Option_values3}
\end{table}

In Table \ref{Option_values3}, we consider also European and American put option prices 
calculated with the Black-Scholes (BS) models.
The BS volatility is chosen equal to the VG volatility $ \sigma^{BS} = (\sigma^2 + \theta^2 \kappa)$.
As expected, the deep OTM (out of the money) prices computed under VG are higher than the corresponding prices computed under BS. This is a consequence of the
shape of the VG density function (\ref{pdf_VG}), which has \emph{heavier tails} than the normal distribution. This means that the probability of a deep OTM option to return in the money, 
is higher if calculated with the VG model than BS, and therefore we get higher option prices.

The Black-Scholes prices are computed using a binomial algorithm. For European options the prices converge to the BS closed formula prices. 
The same values can be obtained using the multinomial algorithm for the VG process and setting $\theta = \kappa = 0$ and $\sigma = \sigma^{BS}$.
Recall that under the Black-Scholes model, the log-returns follow a Brownian motion. 
Looking at the definition of the VG process (\ref{VG_process}), it is evident that when $\theta$ and $\kappa$ are zero, the process becomes a Brownian motion:
$$ X^{VG}_t \underset{\theta,\kappa \to 0}{\to} \sigma W_t. $$
As a consequence, the price process (\ref{S_pprocess}) converges to the Geometric Brownian Motion:
$$ S_t = S_0 e^{(r+\omega)t + X_t} \underset{\theta,\kappa \to 0}{\to} S_0 e^{(r-\frac{1}{2}\sigma^2)t + \sigma W_t} $$
where:
\begin{align*}
 \lim_{\theta,\kappa \to 0} w &= \lim_{\theta,\kappa \to 0} \frac{1}{\kappa} \log(1-\theta \kappa -\frac{1}{2}\sigma^2 \kappa) \\
 &= -\frac{1}{2}\sigma^2.
\end{align*}

\section{Finite difference approximation} \label{sec5}

Consider the VG PIDE (\ref{VG_PIDE}):
\begin{equation}
 \frac{\partial V(t,x)}{\partial t} + (r+\omega) \frac{\partial V(t,x)}{\partial x}
 + \int_{\R} \bigl[ V(t,x+y) - V(t,x) \bigr] \nu(dy) = rV(t,x). 
\end{equation}
We can expand $V(t,x+y)$ using the Taylor formula up to the fourth order:
\begin{align}
 V(t,x+y) &= V(t,x) + \frac{\partial V(t,x)}{\partial x} y + \frac{1}{2} \frac{\partial^2 V(t,x)}{\partial x^2} y^2 \\ \nonumber
  &+ \frac{1}{6} \frac{\partial^3 V(t,x)}{\partial x^3} y^3 + \frac{1}{24} \frac{\partial^4 V(t,x)}{\partial x^4} y^4 
\end{align}
and use the expression for the cumulants (see Appendix \ref{App1}). We denote with $\tilde c_n$ the cumulant evaluated at $t=1:$
\begin{equation}
 \tilde c_n = \int_{\R} y^n \nu(dy). 
\end{equation}
The approximated equation is a fourth order PDE:
\begin{align}
 & \frac{\partial V(t,x)}{\partial t} + (r+\omega+\tilde c_1) \frac{\partial V(t,x)}{\partial x}
 + \frac{1}{2} \tilde c_2 \frac{\partial^2 V(t,x)}{\partial x^2} \\ \nonumber
 &+ \frac{1}{6} \tilde c_3 \frac{\partial^3 V(t,x)}{\partial x^3} + \frac{1}{24} \tilde c_4 \frac{\partial^4 V(t,x)}{\partial x^4} 
 = rV(t,x). 
\end{align}
Consider the variable $x$ in the interval $[x_{min},x_{max}]$ and   
discretize time and space, such that $h = \Delta x = \frac{x_{max} - x_{min}}{N} $ and $\Delta t = \frac{T-t_0}{M}$ 
for $N,M \in \N$.
Using the variables $x_i = x_{min} + ih$ for $i=0,...,N$ and $t_n= t_0 + n\Delta t$ for $n=0,...,M$, we use the short notation
$$V(t_n,x_i) = V^n_i.$$
We can use the following discretization for the time derivative, corresponding to an explicit method:
\begin{equation}
 \frac{\partial V(t,x)}{\partial t} \approx \frac{V^{n+1}_i - V^n_i}{\Delta t}, 
\end{equation}
and the central difference for the spatial derivative:
\begin{align}
 \frac{\partial V(t,x)}{\partial x} \approx& \; \frac{V^{n+1}_{i + h} - V^{n+1}_{i-h}}{2h}, \\ \nonumber 
 \frac{\partial^2 V(t,x)}{\partial x^2} \approx& \; \frac{V^{n+1}_{i + h} + V^{n+1}_{i - h} - 2V^{n+1}_{i}}{h^2}, \\ \nonumber 
 \frac{\partial^3 V(t,x)}{\partial x^3} \approx& \; \frac{ V^{n+1}_{i + 2h} - V^{n+1}_{i - 2h} 
 +2V^{n+1}_{i - h} - 2V^{n+1}_{i + h} }{2h^3}, \\ \nonumber
 \frac{\partial^4 V(t,x)}{\partial x^4} \approx& \; \frac{ V^{n+1}_{i - 2h} + V^{n+1}_{i + 2h} -4V^{n+1}_{i - h} - 
 4V^{n+1}_{i + h} + 6V^{n+1}_{i}}{h^4}. \\ \nonumber
\end{align}
The discretized equation is:
\begin{align}
 &\biggl( \frac{V^{n+1}_i - V^n_i}{\Delta t} \biggr) 
 + (r+\omega+\tilde c_1) \biggl( \frac{V^{n+1}_{i + h} - V^{n+1}_{i-h}}{2h} \biggr)\\ \nonumber
 &+ \frac{1}{2} \tilde c_2 \biggl( \frac{V^{n+1}_{i + h} + V^{n+1}_{i - h} - 2V^{n+1}_{i}}{h^2} \biggr)
 + \frac{1}{6} \tilde c_3 \biggl( \frac{ V^{n+1}_{i + 2h} - V^{n+1}_{i - 2h} 
 +2V^{n+1}_{i - h} - 2V^{n+1}_{i + h} }{2h^3} \biggr) \\ \nonumber		
 &+ \frac{1}{24} \tilde c_4 \biggl( \frac{ V^{n+1}_{i - 2h} + V^{n+1}_{i + 2h} -4V^{n+1}_{i - h} - 
 4V^{n+1}_{i + h} + 6V^{n+1}_{i}}{h^4} \biggr) 
 = rV^n_i .
\end{align}
Rearranging the terms we obtain:
\begin{align}
(1+r\Delta t) V^n_i =& V^{n+1}_{i + h} \biggl[ \frac{(r+\omega+\tilde c_1)\Delta t}{2h} + \frac{\tilde c_2 \Delta t}{2h^2} 
-\frac{\tilde c_3 \Delta t}{6h^3} - \frac{\tilde c_4 \Delta t}{6h^4} \biggr] \\ \nonumber
+& V^{n+1}_{i - h} \biggl[ \frac{-(r+\omega+\tilde c_1)\Delta t}{2h} + \frac{\tilde c_2 \Delta t}{2h^2} 
+\frac{\tilde c_3 \Delta t}{6h^3} - \frac{\tilde c_4 \Delta t}{6h^4} \biggr] \\ \nonumber
+& V^{n+1}_{i + 2h} \biggl[ \frac{\tilde c_3 \Delta t}{12h^3} + \frac{\tilde c_4 \Delta t}{24h^4} \biggr] \\ \nonumber
+& V^{n+1}_{i - 2h} \biggl[ -\frac{\tilde c_3 \Delta t}{12h^3} + \frac{\tilde c_4 \Delta t}{24h^4} \biggr] \\ \nonumber
+& V^{n+1}_{i} \biggl[ 1 - \frac{\tilde c_2 \Delta t}{h^2} + \frac{\tilde c_4 \Delta t}{4h^4} \biggr].
\end{align}
If we rename the coefficients, the equation is:
\begin{equation}\label{prob_coeff}
 (1+r\Delta t) V^n_i = V^{n+1}_{i + h} p_{+h} + V^{n+1}_{i - h} p_{-h} +
  V^{n+1}_{i + 2h} p_{+2h} + V^{n+1}_{i - 2h} p_{-2h} + V^{n+1}_{i} p_{0}. 
\end{equation}
The coefficients can be interpreted as the (risk neutral) transition probabilities for the Markov chain:
%\begin{equation}
\[
X(t_{n+1}) = \begin{dcases}
         X(t_{n}) + 2h & \mbox{with  } \PP(x_i \to x_i +2 h) = p_{+2h} \\
         X(t_{n}) + h & \mbox{with  } \PP(x_i \to x_i + h) = p_{+h} \\
	 X(t_{n})     & \mbox{with  } \PP(x_i \to x_i) = p_{0}   \\
	 X(t_{n}) - h & \mbox{with  } \PP(x_i \to x_i - h) = p_{-h} \\
         X(t_{n}) + 2h & \mbox{with  } \PP(x_i \to x_i -2 h) = p_{-2h}
        \end{dcases} 
\]
%\end{equation}
It is straightforward to verify that $p_{-2h} + p_{-h} + p_{0} + p_{+h} + p_{+2h}  = 1$.
The space step $h$ has to to be chosen in order to satisfy the positivity condition of the transition probabilities, $p_{jh} > 0$ for $j=-2,-1,0,1,2$.
The value of the option in the previous time step is thus the discounted expectation 
under the risk neutral probability measure $\Q$:
\begin{equation}
 V^n_i = \frac{1}{1+r\Delta t} \E^{\Q} \biggl[ V^{n+1}(X(t_{n+1})) \bigg| X(t_n)=x_i \biggr]. 
\end{equation}
Define the increments $\Delta X = X(t_{n+1}) - X(t_n)$.
We check that the local properties for the moments of the Markov chain are satisfied:
\begin{align}
 \mu' =& \; \E[\Delta X] =  \bigl( r+\omega+\tilde c_1 \bigr) \Delta t \\
 \mu'_2 =& \, \E[\Delta X^2] =  \tilde c_2 \Delta t \\
 \mu'_3 =& \, \E[\Delta X^3] =  \bigl( (r+\omega+\tilde c_1)h^2 + \tilde c_3 \bigr) \Delta t \\
 \mu'_4 =& \,\E[\Delta X^4] = \bigl(\tilde c_2 h^2 + \tilde c_4\bigr) \Delta t.
\end{align}
At first order in $\Delta t$ we can calculate the variance, skewness and kurtosis\footnote{Remind that $\mbox{Skew}[X]=\frac{\mu_3}{\mu_2^{3/2}}$
and $\mbox{Kurt}[X]=\frac{\mu_4}{\mu_2^{2}}$, with $\mu_i$ the central i-th moment. Remind also that $\mu_3 = \mu'_3-3\mu'\mu'_2
+2\mu'^3$ and $\mu_4 = \mu'_4 -4\mu'\mu'_3 +6\mu'^2\mu'_2 -3\mu'^4$ } :
\begin{align}
 \mbox{Var}[\Delta X] &\approx \tilde c_2 \Delta t \\
 \mbox{Skew}[\Delta X] &\approx \frac{(r+\omega+\tilde c_1)}{(\tilde c_2)^{3/2}} 
 \frac{h}{\sqrt{\Delta t}} + \frac{\tilde c_3}{(\tilde c_2)^{3/2}} \frac{1}{\sqrt{\Delta t}} \\
 \mbox{Kurt}[\Delta X] &\approx \frac{h^2}{\tilde c_2}
 \frac{1}{\Delta t} + \frac{\tilde c_4}{(\tilde c_2)^2} \frac{1}{\Delta t}.
\end{align}
For $h\to 0$, we confirm that the local variance,
skewness and kurtosis are consistent with their definition in terms of cumulants.

We expect that the transition probabilities in (\ref{probabilities2}) obtained by moment matching, can be recovered from the probabilities in
(\ref{prob_coeff}) obtained by finite difference discretization, for a particular choice of the space step $h$. 
We have not solved this issue yet, and we think it can be a topic for further research. 
%$$ p^{MM} = p_3 = \frac{3+2\bar \kappa}{2(3+\bar \kappa)} \hspace{1em} 
%\mbox{and} \hspace{1em}  p^{PDE} = p_0 = 1 - \frac{\tilde c_2 \Delta t}{h^2} + \frac{\tilde c_4 \Delta t}{4h^4} $$
%\begin{figure}[t!]
% \centering
% \includegraphics[scale=0.5]{p3.png}
 % p3.png: 0x0 pixel, 0dpi, 0.00x0.00 cm, bb=
% \caption{Probabilities $p^{MM}$ and $p^{PDE}$ as functions of the kurtosis $\kappa$.}
% \label{FigP3}
%\end{figure}

\section{Conclusions} \label{sec6}

This article proposes a method to price options using a multinomial method when the underlying price is modeled with a Variance Gamma 
process.
The multinomial method is well known in the literature, see for example 
\cite{Cont}, \cite{KeWe06} ,\cite{YaPr01}, \cite{YaPr03} and \cite{YaPr06}, but in this work we focus the analysis only on the VG process 
and compare our numerical results with other different methods. 

The VG process is approximated by a general jump process that has the same first four cumulants of the original VG process. 
We proved that the multinomial method converges to this approximated process.
We obtained numerical results for European and American options, and compared them with PIDE methods and with results computed within
the Black Scholes framework. 
It turns out that the multinomial method is easier to implement than the finite differences method. The algorithm does not involve any matrix multiplication or matrix
inversion as in the case of implicit/explicit method for PIDEs. This means that the computational time is much smaller. 

We proposed a topic of further investigation in Section \ref{sec5}. The probabilities obtained by discretizing the approximated PDE are
related with the probabilities obtained by moment matching for a particular choice of the space step parameter.  
Another possible topic of further research is the comparison of our results for American options 
with other numerical methods such as the Least Square Monte Carlo (\cite{LoSc01}).

\section*{Acknowledgements}
Our sincere thanks are for the Department of Mathematics of ISEG and CEMAPRE, University of Lisbon, 
\url{http://cemapre.iseg.ulisboa.pt/}.
This research was supported by the European Union in the FP7-PEOPLE-2012-ITN project STRIKE - 
Novel Methods in Computational Finance (304617), and by CEMAPRE
MULTI/00491, financed by FCT/MEC through Portuguese national funds.
We wish also to acknowledge all the members of the STRIKE network, \url{http://www.itn-strike.eu/}.

\appendix
\section{Poisson integration}\label{App1}

A convenient tool to analyze the jumps of a Lévy process is the random
measure of jumps. Within this formalism it is possible to describe jump processes with infinite activity, as the VG process.
The jump process associated to the Lévy process $X_t$ is
defined, for each $0\leq t\leq T$ , by:
\begin{equation}\label{jump}
 \Delta X_t = X_t - X_{t-}
\end{equation}
where $X_{t-} = \lim_{s\uparrow t} X_s $. 
Consider the set $A \in \mathcal{B}(\R \backslash \{ 0 \})$ ,
the \emph{random measure} of the jumps of $X_t$ is defined by:
\begin{align}
 N(t,A)(\omega) &= \# \{ \Delta X_s(\omega) \in A \; : \; 0\leq s \leq t  \} \\
		   &= \sum_{s\leq t} 1_A(\Delta X_s(\omega)) . \nonumber
\end{align} 
This measure counts the number of jumps of size in $A$, up to time $t$.
We say that $A\in \mathcal{B}(\R \backslash \{ 0 \})$ is \emph{bounded below} if $0 \not \in \bar A$ (zero does not belong to the closure of $A$). 
If $A$ is bounded below, then $N(t,A) < \infty$ and is a Poisson process with intensity 
 \begin{equation}
 \nu(A) = \mathbb{E}[N(1,A) ] ,
 \end{equation}
 (see \cite{Applebaum} 
 theorem 2.3.4 and 2.3.5). 
 If $A$ is not bounded below, it is possible to have $\nu(A) = \infty$ and $N(t,A)$ is not a Poisson process because of the accumulation of infinite numbers of
 small jumps.
 The process $N(t,A)$ is called \emph{Poisson random measure}. The Lévy measure corresponds to the intensity 
 of the Poisson measure.
 The \emph{Compensated Poisson random measure} is defined by
 \begin{equation}
  \tilde{N}(t,A) = N(t,A) - t\nu(A), 
 \end{equation}
which is a martingale.

The next step is to define the integration with respect to a random measure.
Following \cite{Applebaum}, let $f:\R \to \R$ be a Borel-measurable
function. For any $A$ bounded below, we define the \emph{Poisson integral} of $f$ as:
 \begin{equation}\label{Poisson_int}
  \int_A f(x) N(t,dx)(\omega) = \sum_{x\in A} f(x) N(t,\{x\})(\omega). 
 \end{equation}
For the case of integration of the identity function, we see that every compound Poisson process can be represented by:
\begin{equation}\label{compound_P}
 X_t = \sum_{s\in [0,t]} \Delta X_s = \int_0^t \int_{\R} x N(dt,dx) = \int_{\R} x N(t,dx).
\end{equation}
We can also define in the same way the \emph{compensated Poisson integral} with respect the compensated Poisson measure.
We can further define:
\begin{equation}
\int_{|x|<1} x \tilde N(t,dx) = \lim_{\epsilon \to 0} \int_{\epsilon < |x| < 1} x \tilde N(t,dx), 
\end{equation}
that represent the compensated sum of small jumps.

We present a last formula for computing the moments of a general compound Poisson process.
Let $f: \R \to \R$ be a measurable function such that $\int_A |f(x)| \nu(dx) < \infty$, and $ X_t = \int_A f(x) N(t,dx)$, the characteristic function of $X_t$ is:
 \begin{align}\label{Exp_poiss}
 \E[ e^{iuX_t} ] &= 
  \E \left[ \exp \left( i u \int_A f(x) N(t,dx)) \right) \right] \\ \nonumber
  &= \exp \left( t \int_{\R} \bigl[ e^{iu f(x)}-1 \bigr] \; \nu \bigl( A \cap f^{-1}(dx)\bigr) \right).
 \end{align}
Assuming that $\E[X_t^n] < \infty$, all the moments can be computed from (\ref{Exp_poiss}) by differentiation using eq:
\begin{equation}\label{moments}
 \E[X_t^{n}] = \frac{1}{i^n} \frac{\partial^n \phi_{X_t}(u)}{\partial u^{n}} \biggr|_{u=0}, \quad \forall n \in \N .
\end{equation}
For the case of $f$ identity function, $A=\R$, and $\nu$ satisfying the integrability conditions, using (\ref{cumulants}) we derive the expression for the cumulants: 
\begin{equation}\label{levy_cumulants}
 c_n = t \int_{\R} x^n \nu(dx). 
\end{equation}
The cumulants of $X_t$ are thus the moments of its Lévy measure. 

\section{Cumulants}\label{App2}

The cumulant generating function $H_{X_t}(u)$ of $X_t$ is defined as the natural logarithm of its characteristic function
(see \cite{Cont}). 
Using the Lévy-Khintchine representation for the characteristic function (\ref{LevyKin}), it is easy to find its relation with
the Lévy symbol:
\begin{align}
H_{X_t}(u) &= \log(\phi_{X_t}(u)) \\ \nonumber
           &= t \eta(u) \\ \nonumber
           &= \sum_{n=1}^{\infty} c_n \frac{(iu)^n}{n!}
\end{align}
The \emph{cumulants} of a Lévy process are thus defined by:
\begin{equation}\label{cumulants}
 c_n = \frac{t}{i^n} \frac{\partial^n \eta(u) }{\partial u^{n}} \biggr|_{u=0} .
\end{equation}
The cumulants are closely related to the central moments $\mu_n$: 
\begin{equation}\label{moment_cumulants}
 \mu_0 = 1,\hspace{1em} \mu_1=0, \hspace{1em} \mu_n=\sum_{k=1}^n \binom{n-1}{k-1}  c_k \mu_{n-k} \hspace{1em} \mbox{for } n>1.  
\end{equation}
For a Poisson process with finite firsts $n$ moments, all the information about the cumulants is contained inside the Lévy measure.
Expand in Taylor series the exponential 
$$ e^{iux} \approx 1 +iux - \frac{u^2x^2}{2} -\frac{iu^3x^3}{3!} +\frac{u^4x^4}{4!} + \dots $$ 
The Lévy symbol from the representation (\ref{Character_Poisson}) becomes:
\begin{align}\label{cumulant_expansion}
 t \eta(u) &= ib'u t + t \int_{\R} (e^{iux} -1) \nu(dx) \\ \nonumber
          &= i\biggl( b-\int_{|x|<1} x \nu(dx) \biggr)ut +iu t\int_{\R} x \nu(dx) - \frac{u^2}{2}t\int_{\R}x^2 \nu(dx) \\ \nonumber 
          & \hspace{2em} -\frac{iu^3}{3!}t\int_{\R} x^3 \nu(dx) +\frac{u^4}{4!}t\int_{\R} x^4 \nu(dx) + \dots\\ \nonumber
	  &= ic_1 u -\frac{c_2u^2}{2} -\frac{ic_3u^3}{3!} + \frac{c_4u^4}{4!} + \dots 
\end{align}
with $c_1= t \bigl( b+\int_{|x|\geq 1} x \nu(dx) \bigr)$.

\bibliographystyle{apalike}

%\bibliographystyle{plain}
%\bibliography{/home/nicola/Documenti/ISEG/Bibliografia/book.bib,/home/nicola/Documenti/ISEG/Bibliografia/fin_math.bib,/home/nicola/Documenti/ISEG/Bibliografia/viscosity.bib,/home/nicola/Documenti/ISEG/Bibliografia/num_meth.bib,/home/nicola/Documenti/ISEG/Bibliografia/phd_thesis.bib,/home/nicola/Documenti/ISEG/Bibliografia/Levy.bib}

\end{document}